\documentclass[11pt, reqno]{amsart}
\usepackage{amsmath}
\usepackage{amsfonts}
\usepackage{comment}
\usepackage{amssymb}
\usepackage{amsthm}
\usepackage[english]{babel}
\usepackage{latexsym}
\usepackage{blindtext}
\usepackage{eufrak}
\usepackage{nicematrix}

\newcommand{\bbN}{{\mathbb{N}}}

\newcommand{\bbR}{{\mathbb{R}}}

\newcommand{\bbE}{{\mathbb{E}}}

\newcommand{\bbC}{{\mathbb{C}}}

\usepackage{bbm}

\newcommand{\beq}{\begin{equation}}

\newcommand{\eeq}{\end{equation}}

\newcommand{\ba}{\begin{align}}

\newcommand{\ea}{\end{align}}

\renewcommand{\Re}{\mathop{\mathrm{Re}}}
\renewcommand{\Im}{\mathop{\mathrm{Im}}}
\newtheorem{remark}{Remark}

\allowdisplaybreaks

\numberwithin{equation}{section}

\usepackage{mathrsfs}
\usepackage{enumitem}
\usepackage{mathtools}
\usepackage[utf8]{inputenc}
\newtheorem{theorem}{Theorem}[section]

\newtheorem{prop}[theorem]{Proposition}

\newtheorem{lemma}[theorem]{Lemma}

\newtheorem{corollary}{Corollary}

\theoremstyle{definition}

\newtheorem{definition}{Definition}

\newtheorem{example}{Example}

\title[Stability of Mesoscopic Fluctuations]{Stability of Mesoscopic Fluctuations of Orthogonal Polynomial ensembles Under Sparse decaying Perturbations}
\author{Daniel Ofner}
\thanks{Institute of Mathematics, The Hebrew University of Jerusalem, Jerusalem, 91904, Israel. E-mail: daniel.ofner@mail.huji.ac.il. Research supported in part by Israel Science Foundation Grant No.\ 1378/20}
\begin{document}
\begin{abstract}
    We study the stability of the mesoscopic fluctuations of  certain orthogonal polynomial ensembles on the real line utilizing the recurrence relation of the associated orthogonal polynomials. We prove that under a sparse enough decaying perturbation of the recurrence coefficients the limiting distribution is stable. 
    As a corollary we prove a mesoscopic central limit theorem (at any scale) for a family of singular continuous measures on $[-2,2]$. 
\end{abstract}

\maketitle
\sloppy
\section{Introduction}
Let $\mu$ be a probability measure on $\mathbb{R}$ with finite moments which is supported on an infinite set. The orthogonal polynomial ensemble (OPE) associated with $\mu$ (OPE$_n\left(\mu\right)$) is a point process of $n$ points in $\mathbb{R}$ whose density function (on $\mathbb{R}^n$) proportional to 

\beq \label{eq:OPE}
\Pi_{1\leq k < j\leq n}\left|\lambda_k-\lambda_j\right|^2\Pi_{1\leq k\leq n}d\mu(\lambda_k).
\eeq

Orthogonal polynomial ensembles arise naturally in probability models, particularly in random matrix theory. For a comprehensive survey on OPE's, see \cite{Konig}. The reason for the name ``OPE" is that this process is closely related to the orthogonal polynomials associated with $\mu$, as we shall see below. We recall the definition of orthogonal polynomials $\left\{p_j\right\}_{j=0}^\infty$ with respect to $\mu$.\\
For any $j\geq0$
$$p_j(x)=\kappa_jx^j+\text{lower order terms}$$ with $\kappa_j>0$, and 
$$\int p_j p_k d\mu =\delta_{kj}.$$ 
These orthogonal polynomials satisfy a three term recurrence relation. 
There exist two sequences of real numbers $\left\{a_n,b_n\right\}_{n=1}^\infty\text{ with}\: a_n>0$ (commonly known as ``Jacobi recurrence coefficients") so that for any $n\in\mathbb{N}$   
\beq \label{eq:recurrence}
xp_n(x)=a_{n+1}p_{n+1}(x)+b_{n+1}p_{n}(x)+a_{n}p_{n-1}(x)
\eeq
and for $n=0$
$$xp_0(x)=a_1p_1(x)+b_1p_0(x).$$
Any probability measure can be associated with a unique sequence of recurrence coefficients, and vice verse- any two bounded sequences $\left\{a_n,b_n\right\}_{n=1}^\infty$ with $a_n>0$, uniquely determine the measure of orthogonality.\\
The connection is made using the Jacobi matrix associated with $\left\{a_n,b_n\right\}_{n=1}^\infty$
$$J=\begin{pmatrix}b_1&a_1&0&0&0&\ldots  \\
a_1&b_2&a_2&0&0&\ldots \\
0&a_2&b_3&a_3&0&\ldots\\
\ddots&\ddots&\ddots&\ddots&\ddots&\ddots\end{pmatrix}$$
and the observation that it is a bounded self adjoint operator, acting on $\ell_2\left(\mathbb{N}\right)$ with a cyclic vector $e_1=(1,0,0,\ldots)^T$. One consequence of the spectral theorem is that the spectral measure of $J$ and $e_1$, is the measure of orthogonality of the set $\{p_j\}_{j=0}^\infty$ which is produced according to \eqref{eq:recurrence}. This leads to the understanding that the recurrence relation encapsulates all the information regarding $\mu$ as well as OPE$_n(\mu$).\\
The approach of analyzing the recurrence relation to study the asymptotics of the related processes has proven beneficial and serves as the main idea behind several works, including \cite{BreuerSineKernel,BreuerDuitsMeso,BreuerDuitsMacro,BreuerOfner,DuitsKozhan,hardy,KVA,Lambert}. The effectiveness of this approach lies in its applicability to studying OPE's associated with non-regular probability measures. Our work employs this approach as well.
Another observation which ties the orthogonal polynomials to $\text{OPE}_n(\mu)$ is 
\beq \label{eq:OPECDkernel}
\Pi_{1\leq k < j\leq n}\left|\lambda_k-\lambda_j\right|^2=Z_n\det \left(K_n\left(\lambda_k,\lambda_j\right)\right)_{0\leq k,j\leq n}
\eeq
where $K_n\left(x,y\right)=\sum_{k=0}^{n-1}p_k(x)p_k(y)$ and $Z_n$ is a normalization constant. $K_n$ is commonly referred to as the "Christoffel-Darboux kernel" (with respect to $\mu$) and it can be seen as the integral kernel of the orthogonal projection (in $L_2\left(\mu\right)$) onto the space $\text{span}\left(1,x,\ldots,x^{n-1}\right),$ a fact that we exploit in Section 2.\\
One major concept in the study of point processes is the empirical measure
$$\nu_n=\frac{1}{n}\sum_{k=1}^n\delta_{x_k}$$
where $\{x_1,\ldots,x_n\}$ are randomly chosen with respect to \eqref{eq:OPE}, which make $\nu_n$ a random normalized counting measure.\\
The aim of this paper is to analyze the stability of the mesoscopic fluctuations of the empirical measure. A key tool in studying the mesoscopic fluctuations is the mesoscopic linear statistic. \\
Let $f:\mathbb{R}\to \mathbb{R}$ be a compactly supported function, let $0<\gamma<1$ and $x_0\in \mathbb{R}$.
Provided that we are given $\left\{x_1,\ldots,x_n\right\}$, randomly chosen according to a given OPE$_n(\mu$), the following random variable  
\beq\label{eq:linearstat}
X^{(n)}_{\gamma,x_0}\left(f\right)=\sum_{k=1}^n f\left(n^\gamma \left(x_k-x_0\right)\right)
\eeq
is called the (mesoscopic) linear statistic (associated with $f$, of a scale $\gamma$ around $x_0$).\\
Note that the function $f\left(n^\gamma\cdot\right)$ is supported on an interval of size $\sim \frac{1}{n^\gamma}$, therefore for an appropriate choice of $x_0$, and for a nice enough $\mu$, we would typically see only $\sim n^{1-\gamma}$ elements out of $\left\{x_1,\ldots,x_n\right\}$.\\
Several studies have explored the mesoscopic regimes for different real line and circular ensembles, including \cite{BreuerDuitsMeso} (Jacobi ensembles),\cite{Bot-Kho1,Bot-Kho2,LambertJohanssonMeso} GOE and GUE, \cite{BergDuits,BreuerOfner,SoshnikovMeso} (Circular ensembles and other compact groups), \cite{he-knowles,lod-simm} (Wigenr matrices) and \cite{BekLod,LambertBeta,landon-sosoe} (G$\beta$E, C$\beta$E) proving central limit theorems and universality results.\\
In order to state our results, we define 
\begin{definition} 
       Let $\left\{n_k\right\}_{k\in\mathbb{N}}$ be an increasing sequence of natural numbers and let $\beta\in (0,1)$, we say that the sequence $\left\{n_k\right\}_{k\in\mathbb{N}}$ is $\beta-$\textit{spaced} if for all $M>0$ $\exists n_0\in \mathbb{N}$ such that $\forall n>n_0$ 
       $$\left|\left[n-Mn^\beta,n+Mn^\beta\right]\cap \left\{n_k\right\}_{k\in\mathbb{N}}\right|\leq 1.$$
\end{definition}
\begin{example} \label{ex:poynomialgrowth}
The sequence $\left\{n_k\right\}_{k\in\mathbb{N}}$ defined by
    $n_k=\left[k^{\frac{1}{1-\beta}+\epsilon}\right]$ is $\beta-$\textit{spaced}, for any $\epsilon>0$.\\
    Indeed, if $\epsilon>0$ and $M>0$, if
    $$n_0-Mn_0^\beta\leq k_0^{\frac{1}{1-\beta}+\epsilon}\leq n_0+Mn_0^\beta$$ 
    then $n_{k_0+1}-n_{k_0}> \left(\frac{1}{1-\beta}+\epsilon\right) k_0^{\frac{\beta}{1-\beta}+\epsilon}$ and 
    \beq
    \begin{split}
    & n_{k_0+1}>k_0^{\frac{1}{1-\beta}+\epsilon}+\left(\frac{1}{1-\beta}+\epsilon\right) k_0^{\frac{\beta}{1-\beta}+\epsilon}\\
    &\quad\qquad=k_0^{\frac{1}{1-\beta}+\epsilon}+\left(\frac{1}{1-\beta}+\epsilon\right)k_0^{(1-\beta)\epsilon}\left(k_0^{\frac{1}{1-\beta}+\epsilon}\right)^\beta\\
    &\quad\qquad\geq n_0-Mn_0^\beta+\left(\frac{1}{1-\beta}+\epsilon\right)k_0^{(1-\beta)\epsilon}\left(n_0-Mn_0^\beta \right)^\beta\\
    &\quad\qquad=n_0+\left(\left(1-\frac{M}{n_0^{1-\beta}}\right)\left(\frac{1}{1-\beta}+\epsilon\right)k_0^{(1-\beta)\epsilon}-M\right)n_0^\beta
    \end{split}
    \eeq
    Provided that $k_0,n_0>>1$, we can assume $\left(1-\frac{M}{n_0^{1-\beta}}\right)\left(\frac{1}{1-\beta}+\epsilon\right)k_0^{(1-\beta)\epsilon}>2M$ and we are done. Note that obviously, this implies that if $n_k$ grows faster than any polynomial (as $k\to \infty$), then $n_k$ is $\beta-$spaced.  
\end{example}
We denote $\mu_0=\frac{\sqrt{4-x^2}}{2\pi}\chi_{[-2,2]}dx$, so the associated recurrence coefficients are $a_n=1$ and $b_n=0$ for any $n\in \mathbb{N}.$ The corresponding Jacobi matrix is  
$$J_0=\begin{pmatrix}0&1&0&0&0&\ldots  \\
1&0&1&0&0&\ldots \\
0&1&0&1&0&\ldots\\
\ddots&\ddots&\ddots&\ddots&\ddots&\ddots\end{pmatrix}$$ 
which is commonly known as the ``free" Jacobi matrix.\\
Let $0<\gamma<\beta<1$, and let $\{\lambda_n\}_{n=1}^\infty$ be a real sequence so that $\lambda_n\xrightarrow[n\to \infty]{}0$. Given a $\beta$-spaced sequence $\{n_k\}_{k=1}^\infty$, we define 
\beq \label{eq:J=J0+V}
J=J_0+V\eeq
where $V$ is the following diagonal matrix
$$V_{n,n}=\begin{cases}\lambda_k\quad n=n_k\\
0 \qquad otherwise\end{cases}.$$ 
Our main theorem is
\begin{theorem} \label{thm:main}
  Let $J_0$ be the free Jacobi matrix. Then for every $0<\gamma<1,\:x_0\in \left(-2,2\right)$ and $\forall\beta\in (\gamma,1)$, if $J$ is as in \eqref{eq:J=J0+V} and $\mu$ is the associated spectral measure of $J$ and $e_1$, then for any $f\in C^1_c\left(\bbR\right)$ and $m\in\bbN$ we have \begin{equation}
    \bbE_{\mu_{0}}\left[\left(X^{(n)}_{\gamma,x_0}(f)-\bbE_{\mu_{0}}X^{(n)}_{\gamma,x_0}(f)\right)^m\right]-\bbE_{\mu}\left[\left(X^{(n)}_{\gamma,x_0}(f)-\bbE_{\mu}X^{(n)}_{\gamma,x_0}(f)\right)^m\right]\xrightarrow[n\to\infty]{}0
\end{equation}
where $\bbE_{\mu_{0}},\bbE_{\mu}$ denote the expectations with respect to $\text{OPE}(\mu_{0}),\text{ OPE}(\mu)$ respectively and $C^1_c\left(\bbR\right)$ denotes the set of real, compactly supported, continuously differentiable functions.  \end{theorem}
\begin{remark}
    We note that for $a,b\in\mathbb{R}$ where $a>0$ the mapping $x\mapsto ax+b$ maps $\mu_0$ to the measure $\mu_{a,b}$ for which $a_n=a$ and $b_n=b$. Hence, Theorem \ref{thm:main} is valid for any measure $\mu$ with constant recurrence coefficients.
\end{remark}
The essence of this theorem is that the limiting distribution of the mesoscopic fluctuations (of a given scale, $0<\gamma<1$) is unaffected by perturbations (considered as perturbations of the recurrence coefficients), provided that these perturbations are ``rare enough". This demonstrates the stability of such fluctuations and shows that the limiting distribution of the mesoscopic linear statistics is universal, meaning that the same asymptotic behavior holds for different point processes.\\
We prove this theorem by analyzing the recurrence coefficients and the corresponding Jacobi matrices. This approach was first utilized in \cite{BreuerDuitsMacro} where Breuer and Duits derived an explicit formula for the cumulants of the macroscopic-scale linear statistics (setting $\gamma=0$ and $x_0=0$ in \eqref{eq:linearstat}) in terms of recurrence coefficients. In a subsequent work \cite{BreuerDuitsMeso} they investigated the mesoscopic fluctuations and their stability, demonstrating that the limiting fluctuations are universal for vanishing perturbations, provided that the decay rate is sufficiently rapid. It is important to emphasize that the vanishing perturbation in Theorem \ref{thm:main} can decay arbitrarily slowly to 0. The condition regarding the convergence rate in \cite{BreuerDuitsMeso} is, in a sense, replaced by the occurrences of the perturbation, i.e.\ $\beta-$spacing.\\
Moreover, a Central Limit Theorem (CLT) was established in \cite{BreuerDuitsMeso} for the mesoscopic scales, proving that for $0<\gamma<1$ and $f\in C_c^1(\mathbb{R})$ 
\beq
X^{(n)}_{\gamma,x_0}(f)-\bbE_{\mu_{0}}X^{(n)}_{\gamma,x_0}(f)\xrightarrow[n\to \infty]{\mathcal{D}} \mathcal{N}\left(0,\sigma_f^2\right)
\eeq
where $\mu_{0}$ and $x_0$ are as in Theorem \ref{thm:main} and
\beq \label{eq:var}
\sigma_f^2=\frac{1}{4\pi^2}\int \int \left(\frac{f(x)-f(y)}{x-y}\right)^2dxdy.\eeq
We remark that the limiting variance is independent of $\gamma$ and $x_0$, and is identical to the mesoscopic limiting variance obtained in \cite{BreuerOfner,LambertJohanssonMeso,Lambert-meso,SoshnikovMeso}. With this CLT in mind, Theorem \ref{thm:main} implies
\begin{corollary} \label{cor:mesoCLT}
    Under the assumptions of Theorem \ref{thm:main} we have
    \beq
X^{(n)}_{\gamma,x_0}(f)-\bbE_{\mu}X^{(n)}_{\gamma,x_0}(f)\xrightarrow[n\to \infty]{\mathcal{D}} \mathcal{N}\left(0,\sigma_f^2\right)
\eeq
where $\sigma_f^2$ is as in \eqref{eq:var}. 
\end{corollary}
Another conclusion of Theorem \ref{thm:main} is 
\begin{theorem}
    There exists a family of singular continuous measures on $[-2,2]$, such that the corresponding OPE's exhibit mesoscopic CLT's (at any scale $0<\gamma<1$). 
\end{theorem}
\begin{proof}
Let $\left\{\lambda_n\right\}_{n=1}^\infty$ be a sequence which is not square summable and $\lambda_n\to 0$ as $n\to \infty$. We define a probability measure $\mu$ on $\mathbb{R}$ through its recurrence relation: $a_n=1$ for any $n\in\bbN$ and 
$$b_n=\begin{cases}\lambda_k\quad n=n_k\\
0 \qquad else\end{cases}$$ 
so that $\frac{n_{k+1}}{n_k}\to \infty$ as $k\to \infty.$ By \cite[Theorem 1.5]{KisLastSimon}, $\mu$ is singular continuous on $[-2,2]$. We fix $\gamma\in (0,1)$, and $\gamma<\beta<1$. Note that $\{n_k\}$ is $\beta$-spaced, hence by Corollary \ref{cor:mesoCLT} we have a mesoscopic CLT (at any scale $0<\gamma<1$). \end{proof}
As a matter of fact, note that we defined $\mu$ such that $a_n=1$ for any $n\in \mathbb{N}$ and $b_n\to 0$ as $n\to \infty$ therefore by \cite[Theorem 1.1]{BreuerDuitsMacro} the macroscopic fluctuations of the empirical measure exhibit Gaussian limit as well. Moreover, if $\{n_k\}$ is sufficiently sparse (for exact definition, see \cite{BreuerSineKernel}), we have by \cite[Theorem 1.1]{BreuerSineKernel}, for any $x\in (-2,2)$ and $a,b\in\mathbb{R}$
$$\frac{K_n\left(x+a/n,x+b/n\right)}{K_n(x,x)}\xrightarrow[n\to \infty]{}\frac{\sin\left(\frac{b-a}{\sqrt{4-x^2}}\right)}{\frac{b-a}{\sqrt{4-x^2}}}.$$

While many results assume the smoothness of $\mu$ to establish universality, this model demonstrates that it is not necessary for macroscopic, mesoscopic, or microscopic scales, showing the effectiveness of considering the recurrence relation.

The remainder of the paper is structured as follows:\\
In Section 2, we present some basic concepts from the theory of determinantal point processes and illustrate how the study of such processes translates into the analysis of the recurrence relation. In particular, we provide the cumulant expression in terms of the Jacobi coefficients and reference a key cumulant comparison principle from \cite{BreuerDuitsMeso}.\\
In Section 3, we start by proving Theorem \ref{thm:main} for a specific choice of functions, utilizing the analysis of cumulants. This part is divided into four major steps, each of which simplifies our problem further. We show that proving Theorem \ref{thm:main} goes through the analysis of traces of Hankel operators. We conclude this section by proving Theorem \ref{thm:main} using a density argument developed in \cite{BreuerDuitsMeso}.

{\bf Acknowledgments}
\hfill\\
We thank Jonathan Breuer for useful discussions.\\
This material is based upon work supported by the Swedish Research Council under grant no. 2021-06594 while the author was in residence at Institut Mittag-Leffler in Djursholm, Sweden during the fall semester of 2024.

\section{Preliminaries}
In this section, we review some properties from the theory of determinantal point processes (DPP). In particular, we introduce the cumulant representation of the linear statistics using the recurrence coefficients and a cumulant comparison tool developed in \cite{BreuerDuitsMeso}. We start with
\subsection{Fredholm determinant identity}
\hfill\\
It is classical (see \cite{SoshnikovDPP}) that for a DPP with kernel $K_n$, for a bounded $g:\mathbb{R}\to\mathbb{R}$ 
\beq \label{eq:fred_det}
\mathbb{E}\left[\Pi_{j=1}^n\left(1+g(x_j)\right)\right]=\det\left(1+gK_n\right)_{L_2\left(\mu\right)}
\eeq 
where the RHS of \eqref{eq:fred_det} is the Fredholm determinant of the operator on $L_2(\mu)$, given by
$$\phi\mapsto g(\cdot)\int \phi(y)K_n(\cdot,y)d\mu(y).$$
As mentioned before, $K_n$ in \eqref{eq:OPECDkernel}, is the integral kernel of the orthogonal projection onto the space of polynomials of degree less than $n$, which in turn, is spanned by $\left\{p_0,\ldots,p_{n-1}\right\}$, meaning that it is unitarily equivalent to the orthogonal projection onto the first $n$ coordinates in $\ell_2\left(\mathbb{N}\right)$, $\mathcal{P}_n$.\\
The Spectral Theorem says that multiplication by $x$ in $L_2(\mu)$ is unitarily equivalent to multiplication by $J$ (where $J$ is the Jacobi matrix associated with $\mu$) in $\ell_2\left(\mathbb{N}\right)$, or in other words, $J$ is the matrix representation of multiplication by $x$ with respect to the basis (of $L_2(\mu)$) obtained by the orthonormal polynomials. In particular, if $g:\mathbb{R}\to \mathbb{R}$ is a continuous function, the operator
$$\phi\mapsto g\phi$$ 
is unitarily equivalent to a multiplication by $g(J)$ in $\ell_2\left(\mathbb{N}\right)$. 
Note that determinants are invariant under unitary mapping, hence
\beq \label{eq:l2_det}
\mathbb{E}\left[\Pi_{j=1}^n\left(1+g(x_j)\right)\right]=\det\left(1+gK_n\right)_{L_2\left(\mu\right)}=\det\left(I+g(J)\mathcal{P}_n\right)_{\ell_2\left(\mathbb{N}\right)}.
\eeq 
Given $g:\mathbb{R}\to \mathbb{R}$, the RHS of \eqref{eq:l2_det} is a function of $J$, that is, a function of the recurrence coefficients, tying the recurrence relation with the OPE$_n\left(\mu\right)$.
\subsection{Cumulant generating function and cumulants in terms of $J$.}
Note that if we plug $g=\exp\left(tf\left(n^\gamma\left(\cdot-x_0\right)\right)\right)-1$ in \eqref{eq:l2_det}, then 
\beq \label{eq:momentgen}
\mathbb{E}\left[e^{tX^{(n)}_{\gamma,x_0}(f)}\right]=\det\left(I+\left(\exp\left(tf\left(n^\gamma\left(J-x_0I\right)\right)\right)-I\right)\mathcal{P}_n\right)_{\ell_2\left(\mathbb{N}\right)}.
\eeq 
The cumulants generating function is
$$\mathcal{K}\left[X^{(n)}_{\gamma,x_0}(f)\right](t)=\log\mathbb{E}\left[e^{tX^{(n)}_{\gamma,x_0}(f)}\right],$$
we write its expansion near $t=0$
$$\mathcal{K}\left[X^{(n)}_{\gamma,x_0}(f)\right](t)=\sum_{j=1}^\infty \mathcal{C}_j^{(n)}\left(X^{(n)}_{\gamma,x_0}(f)\right)t^j$$
then $\mathcal{C}_j^{(n)}\left(X^{(n)}_{\gamma,x_0}(f)\right)$ is called the j'th cumulant of $X^{(n)}_{\gamma,x_0}(f)$. It is classical that the j'th cumulant can be written in terms of the first $j$ moments and vice versa. Convergence in moments is thus equivalent to convergence of cumulants, which means that in order to prove Theorem \ref{thm:main} we need to prove for all $m\geq 2$
\beq
\mathcal{C}^{(n),\mu_{0}}_m\left(X^{(n)}_{\gamma,x_0}(f)\right)-\mathcal{C}^{(n),\mu}_m\left(X^{(n)}_{\gamma,x_0}(f)\right)\xrightarrow[n\to\infty]{}0\eeq
where $\mathcal{C}^{(n),\mu_{0}}_m\left(X^{(n)}_{\gamma,x_0}\right)$ and $\mathcal{C}^{(n),\mu}_m\left(X^{(n)}_{\gamma,x_0}\right)$ denote the m'th cumulants with respect to OPE$_n\left(\mu_{0}\right)$ and OPE$_n\left(\mu\right)$ respectively. \\
In what follows, we briefly review the the method obtained in \cite{BreuerDuitsMacro} of expressing cumulants in terms of $J$.\\ 
The Fredholm determinant of $(1+T)$ for a trace class operator $T$, with $\left\|T\right\|<1$ is just $\det(1+T)=\exp\left(\sum_{j\geq1}\frac{(-1)^{j+1}}{j}\text{Tr}\:T^j\right)$ (see \cite[Section 2]{simontraceclass}).\\
By considering $t$ in a small enough neighborhood of 0, the RHS of \eqref{eq:momentgen} is 
$$\exp\left(\sum_{1\leq j}\frac{(-1)^{j+1}}{j}\text{Tr}\:\left(\left(\exp\left(tf\left(n^\gamma\left(J-x_0I\right)\right)\right)-I\right)\mathcal{P}_n\right)^j\right)$$
therefore in some neighborhood of 0
\beq
\mathcal{K}\left[X^{(n)}_{\gamma,x_0}(f)\right](t)=\sum_{1\leq j}\frac{(-1)^{j+1}}{j}\text{Tr}\:\left(\left(\exp\left(tf\left(n^\gamma\left(J-x_0I\right)\right)\right)-I\right)\mathcal{P}_n\right)^j.
\eeq
By expanding the exponent into a power series (where $t$ is in a small enough neighborhood of 0), we get 
\begin{equation} \label{eq:cumulant_in_J}
\begin{split}
    & \mathcal{C}_m^{(n)}\left(X^{(n)}_{\gamma,x_0}(f)\right)\\
    &\quad =\sum_{j=1}^{m}\frac{(-1)^{j+1}}{j}
\sum_{l_1+\ldots+l_j=m, l_i \geq 1}\frac{\textrm{Tr}\left(f_{n,\gamma,x_0}(J)\right)^{l_1}\mathcal{P}_n\cdots
\left(f_{n,\gamma,x_0}(J)\right)^{l_j}\mathcal{P}_n -\textrm{Tr}\left(f_{n,\gamma,x_0}(J)\right)^m
\mathcal{P}_n}{l_1! \cdots l_j!}
\end{split}
\end{equation}
where $f_{n,\gamma,x_0}(J)=f\left(n^\gamma\left(J-x_0I\right)\right).$ This is an explicit expression of the cumulant in terms of recurrence coefficients. It is worth mentioning that $\mathcal{C}^{(n),\mu_{0}}_m\left(X^{(n)}_{\gamma,x_0}\right)$ is just the RHS of \eqref{eq:cumulant_in_J} with $J_{0}$ instead of $J$, thus
\beq \label{eq:cumu_in_f_to_J}
\mathcal{C}_m^{(n),\mu_{0}}\left(X^{(n)}_{\gamma,x_0}(f)\right)-\mathcal{C}_m^{(n),\mu}\left(X^{(n)}_{\gamma,x_0}(f)\right)=\mathcal{C}_m^{(n)}\left(f_{n,\gamma,x_0}(J_{0})\right)-\mathcal{C}_m^{(n)}\left(f_{n,\gamma,x_0}(J)\right).\eeq
\subsection{Comparison principle of cumulants.}
\hfill\\
For an operator $T$, we write
\begin{equation*}
\begin{split}
    & \mathcal{C}_m^{(n)}\left(T\right)\\
    &\quad =\sum_{j=1}^{m}\frac{(-1)^{j+1}}{j}
\sum_{l_1+\ldots+l_j=m, l_i \geq 1}\frac{\textrm{Tr}\left(T\right)^{l_1}\mathcal{P}_n\cdots
\left(T\right)^{l_j}\mathcal{P}_n -\textrm{Tr}\left(T\right)^m
\mathcal{P}_n}{l_1! \cdots l_j!}
\end{split}
\end{equation*}
A key tool in our work is the following comparison principle, proved in \cite{BreuerDuitsMeso}.
\begin{prop} \label{thm:comp_cumulants}
    \cite[Theorem 3.1]{BreuerDuitsMeso}
    Let $0<\gamma<1$ and let $\{ \mathcal T^{(n)} \}_{n=1}^\infty$ $\{ \mathcal S^{(n)} \}_{n=1}^\infty$ be two uniformly bounded (in operator norm) matrix sequences satisfying $\forall j,k\in\mathbb{N}$
$$\left|\mathcal T^{(n)}_{j,k}\right|,\left|\mathcal S^{(n)}_{j,k}\right|\leq Ce^{-\widetilde{d}\frac{|j-k|}{n^\gamma}}$$
for some $C,\widetilde{d}>0$.

For real numbers $\ell_1<\ell_2$ we define
$$P_{\ell_1,\ell_2}=\mathcal P_{[\ell_2]+1}-\mathcal P_{[\ell_1]-1}$$ $($recall $\mathcal P_n$ is the orthogonal projection in $\ell_2\left(\mathbb{N}\right)$ onto the first $n$ coordinates$)$.

Then, for any $\beta'>\gamma$ and any $m \geq 2$,
\begin{equation} \label{eq:almost-band}
\begin{split}
& \left|\mathcal C_m^{(n)} \left(\mathcal T^{(n)} \right) - \mathcal C_m^{(n)} \left( \mathcal S^{(n)} \right) \right| \\
&\quad \leq C(m,\beta') \left\| P_{n-2mn^{\beta'},n+2mn^{\beta'} } \left( \mathcal T^{(n)} - \mathcal S^{(n)}\right)
P_{n-2mn^{\beta'},n+2mn^{\beta'}} \right\|_1+o(1)
\end{split}
\end{equation}
as $n \rightarrow \infty$, where $C(m,\beta')$ is a constant that depends on $m$ and $\beta'$ but is independent of $n$.
\end{prop}

\section{Mesoscopic stability- proof of Theorem \ref{thm:main}}
In this section, we start by proving Theorem \ref{thm:main} for a specific class of functions. We start with a preliminary result regarding function of the form 
\beq \label{eq:f_poissonkernel}
f(x)=\sum_{j=1}^N\frac{c_j}{x-\eta_j}\eeq
where $\eta_1\ldots,\eta_N\in\mathbb{C}$ so that $\text{Im}\eta_j\not=0$ and $c_1\ldots,c_N\in\mathbb{R}$. The reason for dealing with this type of function will become clear toward the end of this section, as we use it as a family of 'test' functions suitable for studying mesoscopic fluctuations. We prove 
\begin{theorem} \label{thm:premain} 
Let $0<\gamma<\beta<1$ and let $\mu_0,\mu,J_0,J$ be as in Theorem \ref{thm:main}. Let $f$ be as in \eqref{eq:f_poissonkernel}, let $x_0\in(-2,2)$, then for any $m\geq2$
\beq\label{cumulantsto0}
\mathcal{C}^{(n),\mu_0}_m\left(X^{(n)}_{\gamma,x_0}(f)\right)-\mathcal{C}^{(n),\mu}_m\left(X^{(n)}_{\gamma,x_0}(f)\right)\xrightarrow[n\to\infty]{}0.\eeq
\end{theorem}
We split the proof of Theorem \ref{thm:premain} into four steps, each of which simplifies the analysis of the cumulants.
\subsection{\textbf{Step 0: Mesoscopic fluctuations in terms of Jacobi matrices}}
\hfill\\
Our initial observation is \eqref{eq:cumu_in_f_to_J}  $$\mathcal{C}^{(n),\mu_0}_m\left(X^{(n)}_{\gamma,x_0}(f)\right)-\mathcal{C}^{(n),\mu}_m\left(X^{(n)}_{\gamma,x_0}(f)\right)=\mathcal{C}_m^{(n)}\left(f_{n,\gamma,x_0}(J_0)\right)-\mathcal{C}_m^{(n)}\left(f_{n,\gamma,x_0}(J)\right)$$
so provided that $\left\{f_{n,\gamma,x_0}(J_0)\right\}_n,\left\{f_{n,\gamma,x_0}(J_0)\right\}_n$ satisfy the conditions of Proposition \ref{thm:comp_cumulants} (we prove it below), we have for $1>\beta'>\gamma$
\beq
\begin{split}
    & \left|\mathcal{C}^{(n),\mu_0}_m\left(X^{(n)}_{\gamma,x_0}(f)\right)-\mathcal{C}^{(n),\mu}_m\left(X^{(n)}_{\gamma,x_0}(f)\right)\right|\\
    &\quad \leq C(m,\beta')\left\|P_{n-2mn^{\beta'},n+2mn^{\beta'}}\left(f_{n,\gamma,x_0}(J_0)-f_{n,\gamma,x_0}(J_0)\right)P_{n-2mn^{\beta'},n+2mn^{\beta'}}\right\|_1+o(1).
\end{split}
\eeq
Note that $\left(n^\gamma(J-x_0)-\eta\right)^{-1}=\frac{1}{n^\gamma}\left(J-\left(x_0+\frac{\eta}{n^\gamma}\right)\right)^{-1}$, 
so by the triangle inequality 
\begin{equation} \label{eq:diff_of_cumulants_linearcomb}
\begin{split}
    & \left|\mathcal{C}^{(n),\mu_0}_m\left(X^{(n)}_{\gamma,x_0}(f)\right)-\mathcal{C}^{(n),\mu}_m\left(X^{(n)}_{\gamma,x_0}(f)\right)\right| \\
    &\quad \leq \frac{C(m,\beta')}{n^\gamma}\sum_{j=1}^N|c_j|\left\|P_{n-2mn^{\beta'},n+2mn^{\beta'}}\left(\left(J_0-z_{n,j}\right)^{-1}-\left(J-z_{n,j}\right)^{-1}\right)P_{n-2mn^{\beta'},n+2mn^{\beta'}}\right\|_1
\end{split}
\end{equation}
where $z_{n,j}=x_0+\frac{\eta_j}{n^\gamma}$ and $C(m,\beta')$ is a positive constant which depends solely on $\beta'$ and $m$.\\
To see why the condition of Proposition \ref{thm:comp_cumulants} are satisfied, using the self adjointness of $J$ and $J_0$, we have for $z\in \mathbb{C}\backslash\mathbb{R}$, $\|(J-z)^{-1}\|,\|(J_0-z)^{-1}\|\leq 1/d(z,\mathbb{R})$ so if $z=z_{n,j}=x_0+\frac{\eta_j}{n^\gamma}$, then $\|(J-z_{n,j})^{-1}\|,\|(J_0-z_{n,j})^{-1}\|=\mathcal{O}\left(n^\gamma\right)$ as $n\to \infty$, so
both $\frac{1}{n^\gamma}(J-z_{n,j})^{-1}$ and $\frac{1}{n^\gamma}(J_0-z_{n,j})^{-1}$ are uniformly bounded (in n).
The second condition is satisfied due to standard estimates of the resolvent of Jacobi matrix, often referred to as the Combes-Thomas estimates (see proof in \cite[Section 2]{BreuerDuitsMeso} and references therein), we can obtain $d_j',D_j'>0$ so that for any $r,l\in \bbN$
\beq
\left|\left(J-z_{n,j}\right)^{-1}_{r,l}\right|,\left|\left(J_0-z_{n,j}\right)^{-1}_{r,l}\right|\leq n^\gamma \cdot D_j'\cdot \exp\left(-\frac{d_j'\cdot|r-l|}{n^\gamma}\right)
\eeq
so by taking $D'=\max\left\{D'_j\right\}_{j=1}^N$ and $d'=\min\left\{d'_j\right\}_{j=1}^N$ we have
\beq
\left|\left(J-z_{n,j}\right)^{-1}_{r,l}\right|,\left|\left(J_0-z_{n,j}\right)^{-1}_{r,l}\right|\leq n^\gamma \cdot D'\cdot \exp\left(-\frac{d'\cdot|r-l|}{n^\gamma}\right).
\eeq
We prove that if $\gamma<\beta'<\beta$ then 
\beq \label{eq:onetermgoesto0}
\frac{1}{n^\gamma}\left\|P_{n-2mn^{\beta'},n+2mn^{\beta'}}\left(\left(J_0-z_n\right)^{-1}-\left(J-z_n\right)^{-1}\right)P_{n-2mn^{\beta'},n+2mn^{\beta'}}\right\|_1\xrightarrow[n\to \infty]{}0\eeq
where $z_{n}=x_0+\frac{\eta}{n^\gamma}$, which is the case $N=1,c_1=1,\eta_1=\eta$, which by \eqref{eq:diff_of_cumulants_linearcomb}, is sufficient to prove Theorem \ref{thm:premain}.\\
We start with
\subsection{Step 1: Decoupling.}
\hfill\\
Let $\gamma<\beta'<\beta<1$. For a given $n\in\bbN$ we define the matrix $$\left(H^{n,\beta}_0\right)_{m,l}=\begin{cases}
    0\quad m-1=l=[n-2mn^\beta] \quad or \quad m=l-1=[n-2mn^\beta]\\
    0\quad m-1=l=[n+2mn^\beta] \quad or \quad m=l-1=[n+2mn^\beta]\\
    \left(J_0\right)_{m,l} \qquad otherwise
\end{cases}$$
and 
$$\left(H^{n,\beta}\right)_{m,l}=\begin{cases}
    0\quad m-1=l=[n-2mn^\beta] \quad or \quad m=l-1=[n-2mn^\beta]\\
    0\quad m-1=l=[n+2mn^\beta] \quad or \quad m=l-1=[n+2mn^\beta]\\
    \left(J\right)_{m,l} \qquad otherwise
\end{cases}$$
In words, by putting zeros at the appropriate spots, we deform $J$, ($J_0$ respectively) into a three block, tridiagonal matrix, making it a direct sum of three Jacobi matrices.   
\begin{lemma} \label{lemma:exp-decay} For any $j,k\in\mathbb{N}$ we have 
$$\left|\left(J_0-z_n\right)^{-1}_{j,k}-\left(H^{n,\beta}_0-z_n\right)_{j,k}^{-1}\right|\leq n^{2\gamma} \cdot D\cdot \exp\left(-\frac{d\cdot\Theta_{n,\beta}(j,k)}{n^\gamma}\right)$$ and $$\left|\left(J-z_n\right)^{-1}_{j,k}-\left(H^{n,\beta}-z_n\right)_{j,k}^{-1}\right|\leq n^{2\gamma}\cdot D\cdot \exp\left(-\frac{d\cdot\Theta_{n,\beta}(j,k)}{n^\gamma}\right).$$ 
where $d,D>0$ are real constants and 
$$\Theta_{n,\beta}(j,k)=\min\left\{\left|n-2mn^\beta-j\right|,\left|n-2mn^\beta-k \right|,\left|n+2mn^\beta-j\right|,\left|n+2mn^\beta-k\right|\right\}.$$
\begin{proof}
By the resolvent formula $$\left(J_0-z_n\right)^{-1}_{j,k}-\left(H^{n,\beta}_0-z_n\right)_{j,k}^{-1}=\left(\left(J_0-z_n\right)^{-1}\left(H^{n,\beta}_0-J_0\right)\left(H^{n,\beta}_0-z_n\right)^{-1}\right)_{j,k}.$$
we compute
\begin{equation}\label{eq:diff_of_resolvents}
\begin{split}
    &\left(\left(J_0-z_n\right)^{-1}\left(H^{n,\beta}_0-J_0\right)\left(H^{n,\beta}_0-z_n\right)^{-1}\right)_{j,k}\\
    &\quad= \sum_{l=1}^\infty\sum_{m=1}^\infty\left(J_0-z_n\right)^{-1}_{j,m}\left(H^{n,\beta}_0-J_0\right)_{m,l}\left(H^{n,\beta}_0-z_n\right)^{-1}_{l,k}\\
    & \quad= \left(J_0-z_n\right)^{-1}_{j,[n-2mn^\beta]}\left(H^{n,\beta}_0-z_n\right)^{-1}_{[n-2mn^\beta]+1,k}+\left(J_0-z_n\right)^{-1}_{j,[n-2mn^\beta]+1}\left(H^{n,\beta}_0-z_n\right)^{-1}_{[n-2mn^\beta],k}\\
    & \qquad+\left(J_0-z_n\right)^{-1}_{j,[n+2mn^\beta]}\left(H^{n,\beta}_0-z_n\right)^{-1}_{[n+2mn^\beta]+1,k}+\left(J_0-z_n\right)^{-1}_{j,[n+2mn^\beta]+1}\left(H^{n,\beta}_0-z_n\right)^{-1}_{[n+2mn^\beta],k}.
\end{split}
\end{equation}
While $ H^{n,\beta}_0$ is not a Jacobi matrix (it has 4 off diagonal zero terms), it can be seen as a direct sum of 3 Jacobi matrices (2 finite and one infinite). The resolvent of $ H^{n,\beta}_0$ is thus a direct sum of 3 resolvents of 3 Jacobi matrices, thus by the Combes-Thomas estimates, for any $j,k\in \bbN$
\beq
\left|\left(J_0-z_n\right)^{-1}_{j,k}\right|,\left|\left(H^{n,\beta}_0-z_n\right)^{-1}_{j,k}\right|\leq n^\gamma \cdot D'\cdot \exp\left(-\frac{d'\cdot|j-k|}{n^\gamma}\right).
\eeq
and therefore we can conclude by plugging this estimate into \eqref{eq:diff_of_resolvents} that there exist $d,D>0$ so that
\beq 
\left|\left(J_0-z_n\right)^{-1}_{j,k}-\left(H^{n,\beta}_0-z_n\right)_{j,k}^{-1}\right|\leq n^{2\gamma} \cdot D\cdot \exp\left(-\frac{d\cdot\Theta_{n,\beta}(j,k)}{n^\gamma}\right)
\eeq
the proof for $$\left|\left(J-z_n\right)^{-1}_{j,k}-\left(H^{n,\beta}-z_n\right)_{j,k}^{-1}\right|$$ is exactly the same. 
\end{proof}
\begin{remark}
    Later, we will estimate $\left(J_0-z_n\right)^{-1}$ and $\left(H_0^{n,\beta}-z_n\right)^{-1}$ entry-wise (see Step 2). One can see that Lemma \ref{lemma:exp-decay} can be improved, so that the $n^{2\gamma}$ factor is redundant.   
\end{remark}
\end{lemma}
This lemma leads us to the following result
\begin{prop} \label{CumulantJvsCumulantH}
As $n\to \infty$ we have
\beq
\begin{split}
  &\left\|P_{n-2mn^{\beta'},n+2mn^{\beta'}}\left(\left(J_0-z_n\right)^{-1}-\left(J-z_n\right)^{-1}\right)P_{n-2mn^{\beta'},n+2mn^{\beta'}}\right\|_1\\
  &\quad = \left\|P_{n-2mn^{\beta'},n+2mn^{\beta'}}\left(\left(H^{n,\beta}_0-z_n\right)^{-1}-\left(H^{n,\beta}-z_n\right)^{-1}\right)P_{n-2mn^{\beta'},n+2mn^{\beta'}}\right\|_1+o(1)
\end{split}
\eeq 
\end{prop}
\begin{proof} 
By the triangle inequality we have
\beq
\begin{split}
  &\left\|P_{n-2mn^{\beta'},n+2mn^{\beta'}}\left(\left(J_0-z_n\right)^{-1}-\left(J-z_n\right)^{-1}\right)P_{n-2mn^{\beta'},n+2mn^{\beta'}}\right\|_1\\
  &\quad \leq \left\|P_{n-2mn^{\beta'},n+2mn^{\beta'}}\left(\left(H^{n,\beta}_0-z_n\right)^{-1}-\left(J_0-z_n\right)^{-1}\right)P_{n-2mn^{\beta'},n+2mn^{\beta'}}\right\|_1\\
  &\qquad +\left\|P_{n-2mn^{\beta'},n+2mn^{\beta'}}\left(\left(H^{n,\beta}_0-z_n\right)^{-1}-\left(H^{n,\beta}-z_n\right)^{-1}\right)P_{n-2mn^{\beta'},n+2mn^{\beta'}}\right\|_1\\
  & \qquad + \left\|P_{n-2mn^{\beta'},n+2mn^{\beta'}}\left(\left(J-z_n\right)^{-1}-\left(H^{n,\beta}-z_n\right)^{-1}\right)P_{n-2mn^{\beta'},n+2mn^{\beta'}}\right\|_1.
\end{split}
\eeq 
We note that 
\beq \begin{split}
   & \left\|P_{n-2mn^{\beta'},n+2mn^{\beta'}}\left(\left(H^{n,\beta}_0-z_n\right)^{-1}-\left(J_0-z_n\right)^{-1}\right)P_{n-2mn^{\beta'},n+2mn^{\beta'}}\right\|_1\\
   &\quad \leq \sum _{n-2mn^{\beta'}\leq j,k \leq n+2mn^{\beta'}}\left|\left(J_0-z_n\right)^{-1}_{j,k}-\left(H^{n,\beta}_0-z_n\right)_{j,k}^{-1}\right|.
\end{split} \eeq
By Lemma \ref{lemma:exp-decay} we see that 
\beq \begin{split}
& \sum _{n-2mn^{\beta'}\leq j,k \leq n+2mn^{\beta'}}\left|\left(J_0-z_n\right)^{-1}_{j,k}-\left(H^{n,\beta}_0-z_n\right)_{j,k}^{-1}\right|\\
&\quad \sum _{n-2mn^{\beta'}\leq j,k \leq n+2mn^{\beta'}} n^{2\gamma} \cdot D\cdot \exp\left(-\frac{d\cdot\Theta_{n,\beta}(j,k)}{n^\gamma}\right).
\end{split} \eeq
Note that $\beta'<\beta$ and therefore for any $n-2mn^{\beta'}\leq j,k \leq n+2mn^{\beta'}$ we have
$$2m\left(n^\beta-n^{\beta'}\right)\leq\Theta_{n,\beta}(j,k)$$ which means 
\beq \begin{split}
   & \left\|P_{n-2mn^{\beta'},n+2mn^{\beta'}}\left(\left(H^{n,\beta}_0-z_n\right)^{-1}-\left(J_0-z_n\right)^{-1}\right)P_{n-2mn^{\beta'},n+2mn^{\beta'}}\right\|_1\\
   &\quad \leq \sum _{n-2mn^{\beta'}\leq j,k \leq n+2mn^{\beta'}} n^{2\gamma} \cdot D\cdot \exp\left(-\frac{2dm\left(n^\beta-n^{\beta'}\right)}{n^\gamma}\right)\\
   & \qquad \leq 4m^2n^{2\beta}\cdot n^{2\gamma} \cdot D\cdot \exp\left(-\frac{2dm\left(n^\beta-n^{\beta'}\right)}{n^\gamma}\right)
\end{split} \eeq
which vanishes as $n\to \infty$. Therefore $$\left\|P_{n-2mn^{\beta'},n+2mn^{\beta'}}\left(\left(\left(H^{n,\beta}_0-z_n\right)^{-1}-\left(J_0-z_n\right)^{-1}\right)\right)P_{n-2mn^{\beta'},n+2mn^{\beta'}}\right\|_1=o(1)$$ as $n\to \infty$. Similarly $$\left\|P_{n-2mn^{\beta'},n+2mn^{\beta'}}\left(\left(H^{n,\beta}-z_n\right)^{-1}-\left(J-z_n\right)^{-1}\right)P_{n-2mn^{\beta'},n+2mn^{\beta'}}\right\|_1=o(1)$$ as $n\to \infty$.\\
This completes the proof.\\ 
\end{proof}
\begin{remark}\label{remark} As a side result of the latter proposition, we obtain a fact that will come in handy later in estimating the entries of $\left(H^{n,\beta}_0-z_n\right)^{-1}$.\\
For any $n-2mn^{\beta'}\leq j,k \leq n+2mn^{\beta'}$ we have 
\beq 
\left|\left(J_0-z_n\right)^{-1}_{j,k}-\left(H^{n,\beta}_0-z_n\right)_{j,k}^{-1}\right|\leq  n^{2\gamma} \cdot D\cdot \exp\left(-\frac{2dm\left(n^\beta-n^{\beta'}\right)}{n^\gamma}\right).
\eeq
\end{remark}
After proving all the above, we see that  in order to prove \eqref{eq:onetermgoesto0} we are left with showing $$\frac{1}{n^\gamma}\left\|P_{n-2mn^{\beta'},n+2mn^{\beta'}}\left(\left(H^{n,\beta}_0-z_n\right)^{-1}-\left(H^{n,\beta}-z_n\right)^{-1}\right)P_{n-2mn^{\beta'},n+2mn^{\beta'}}\right\|_1\xrightarrow[n\to\infty]{}0.$$ The main reason for dealing with $H^{n,\beta}_0$ and $H^{n,\beta}$ rather than $J_0$ and $J$ is the fact that both $H^{n,\beta}_0$ and $H^{n,\beta}$ are 3-block matrices, and it holds that \beq \begin{split}
    & P_{n-2mn^{\beta},n+2mn^{\beta}}\left(\left(H^{n,\beta}_0-z_n\right)^{-1}-\left(H^{n,\beta}-z_n\right)^{-1}\right)P_{n-2mn^{\beta},n+2mn^{\beta}}\\
    &\quad=\left(P_{n-2mn^{\beta},n+2mn^{\beta}}H^{n,\beta}_0P_{n-2mn^{\beta},n+2mn^{\beta}}-z_n\right)^{-1}\\
    &\qquad-\left(P_{n-2mn^{\beta},n+2mn^{\beta}}H^{n,\beta}P_{n-2mn^{\beta},n+2mn^{\beta}}-z_n\right)^{-1}.
\end{split}\eeq
For the sake of convenience we denote $$\Tilde{H}^{n,\beta}_0=P_{n-2mn^{\beta},n+2mn^{\beta}}H^{n,\beta}_0P_{n-2mn^{\beta},n+2mn^{\beta}}$$ and $$\Tilde{H}^{n,\beta}=P_{n-2mn^{\beta},n+2mn^{\beta}}H^{n,\beta}P_{n-2mn^{\beta},n+2mn^{\beta}}.$$
Recall that $V$ was defined in a way that at most only one non-zero entry exists in the range $[n-2mn^\beta,n+2mn^\beta]$ (for $n$ large enough).\\
Assuming that $n-2mn^\beta\leq r=n_k\leq n+2mn^\beta$ with $V_{r,r}=\lambda_k$,
We see that by the resolvent formula \beq \label{eq:Hbeta}
\begin{split}
    &\left(\Tilde{H}^{n,\beta}_0-z_n\right)^{-1}_{j,l}-\left(\Tilde{H}^{n,\beta}-z_n\right)^{-1}_{j,l}\\
    &\quad = \lambda_k \left(\Tilde{H}^{n,\beta}_0-z_n\right)^{-1}_{j,r}\cdot \left(\Tilde{H}^{n,\beta}-z_n\right)^{-1}_{r,l}.
    \end{split}\eeq
Taking $j=r$ in \eqref{eq:Hbeta} and plugging it back to \eqref{eq:Hbeta}, we have for\\
$n-2mn^{\beta}\leq j,l\leq n+2mn^{\beta}$
\beq \label{eq:H0-H} \begin{split}
    &\left(\Tilde{H}^{n,\beta}_0-z_n\right)^{-1}_{j,l}-\left(\Tilde{H}^{n,\beta}-z_n\right)^{-1}_{j,l}\\
    &\quad= \lambda_k \cdot \frac{\left(\Tilde{H}^{n,\beta}_0-z_n\right)^{-1}_{j,r}\cdot \left(\Tilde{H}^{n,\beta}_0-z_n\right)^{-1}_{r,l}}{1+\lambda_k\cdot\left(\Tilde{H}^{n,\beta}_0-z_n\right)^{-1}_{r,r}}
\end{split}
\eeq
Using Remark \ref{remark}, we may see that for $n-2mn^{\beta'}\leq j,l\leq n+2mn^{\beta'}$  
\beq  \begin{split}
    &\left(\Tilde{H}^{n,\beta}_0-z_n\right)^{-1}_{j,l}-\left(\Tilde{H}^{n,\beta}-z_n\right)^{-1}_{j,l}\\
    &\quad= \lambda_k \cdot \frac{\left(J_0-z_n\right)^{-1}_{j,r}\cdot \left(J_0-z_n\right)^{-1}_{r,l}}{1+\lambda_k\cdot\left(J_0-z_n\right)^{-1}_{r,r}}+r_{j,l}(n,\beta,\beta')
\end{split}
\eeq
with $$r_{j,l}(n,\beta,\beta')=\mathcal{O}\left(n^{2\gamma}\exp\left(-d\cdot \left(n^{\beta-\gamma}-n^{\beta'-\gamma}\right)\right)\right)$$ 
as $n\to \infty$.
We now have \beq \begin{split}
    &\left\|P_{n-2mn^{\beta'},n+2mn^{\beta'}}\left(\left(H^{n,\beta}_0-z_n\right)^{-1}-\left(H^{n,\beta}-z_n\right)^{-1}\right)P_{n-2mn^{\beta'},n+2mn^{\beta'}}\right\|_1 \\
     & \quad =\left\|P_{n-2mn^{\beta'},n+2mn^{\beta'}}\left(\left(\Tilde{H}^{n,\beta}_0-z_n\right)^{-1}-\left(\Tilde{H}^{n,\beta}-z_n\right)^{-1}\right)P_{n-2mn^{\beta'},n+2mn^{\beta'}}\right\|_1\\
     &\qquad\leq \left\|G_m^{(n),\beta'}(z_n)\right\|_1+\left\|\left(r_{j,l}\left(n,\beta,\beta'\right)\right)_{n-2mn^{\beta'}\leq j,l\leq n+2mn^{\beta'}}\right\|_1
\end{split}
\eeq
where $G_m^{(n),\beta'}(z_n)$ is the matrix 
$$G_m^{(n),\beta'}(z_n)_{j,l}=\begin{cases}
    \lambda_k \cdot \frac{\left(J_0-z_n\right)^{-1}_{j,r}\cdot \left(J_0-z_n\right)^{-1}_{r,l}}{1+\lambda_k\cdot\left(J_0-z_n\right)^{-1}_{r,r}} \quad  n-2mn^{\beta'}\leq j,l\leq n+2mn^{\beta'}\\
    0 \qquad otherwise
\end{cases}
$$ and $n-2mn^{\beta}\leq n_k=r\leq n+2mn^{\beta} $
\begin{lemma} \label{lemmarjk}
With $\left(r_{j,k}\left(n,\beta,\beta'\right)\right)_{n-2mn^{\beta'}\leq j,l\leq n+2mn^{\beta'}}$ as above, we have $$\left\|\left(r_{j,l}\left(n,\beta,\beta'\right)\right)_{n-2mn^{\beta'}\leq j,l\leq n+2mn^{\beta'}}\right\|_1\xrightarrow[n\to\infty]{}0$$
\end{lemma}
\begin{proof}
 It is straightforward that   
\beq\begin{split}
    &\left\|\left(r_{j,l}\left(n,\beta,\beta'\right)\right)_{n-2mn^{\beta'}\leq j,l\leq n+2mn^{\beta'}}\right\|_1\\
    &\quad\leq \sum_{n-2mn^{\beta'}\leq j,l\leq n+2mn^{\beta'}}\left|r_{j,l}\left(n,\beta,\beta'\right)\right|\\
    &\qquad\leq 4m^2n^{2\beta'}\mathcal{O}\left(n^{2\gamma}\exp\left(-d\cdot \left(n^{\beta-\gamma}-n^{\beta'-\gamma}\right)\right)\right)\\
\end{split}
\eeq 
which vanishes as $n\to \infty$.
\end{proof}
To summarize all the above, we have shown that 
\beq
\begin{split}
    &\left\|P_{n-2mn^{\beta'},n+2mn^{\beta'}}\left(\left(J_0-z_n\right)^{-1}-\left(J-z_n\right)^{-1}\right)P_{n-2mn^{\beta'},n+2mn^{\beta'}}\right\|_1\\
    &\quad=\left\|P_{n-2mn^{\beta'},n+2mn^{\beta'}}\left(\left(H^{n,\beta}_0-z_n\right)^{-1}-\left(H^{n,\beta}-z_n\right)^{-1}\right)P_{n-2mn^{\beta'},n+2mn^{\beta'}}\right\|_1+o(1)\\
    &\qquad\leq \left\|G_m^{(n),\beta'}(z_n)\right\|_1+o(1)
\end{split}
\eeq
as $n\to \infty$.\\
We are now thus left with proving $$ \frac{1}{n^\gamma}\left\|G_m^{(n),\beta'}(z_n)\right\|_1\xrightarrow[n\to \infty]{}0$$ in order to prove \eqref{eq:onetermgoesto0}.\\
\subsection{\textbf{Step 2: Entry-wise estimations}}
\hfill\\
We start by explicitly computing the entries of $G_m^{(n),\beta'}(z_n)$.\\
It is known, and not hard to verify that 
\beq
\left(J_0-z_n\right)^{-1}_{j,k}=\frac{\phi(z_n)^{|k-j|}-\phi(z_n)^{k+j}}{\phi(z_n)-1/\phi(z_n)}
\eeq
here $\phi(\zeta)=\frac{\zeta-\sqrt{\zeta^2-4}}{2}$, where the square root is taken such that $\zeta \to \phi(\zeta)$ is analytic in $\bbC \setminus[-2,2]$ and $\phi(\zeta)= \mathcal O(1/\zeta)$ as $\zeta \to \infty$. Note that this map maps  $\bbC \setminus[-2,2]$ onto  the open unit disk and hence $|\phi(z_n)|<1$.\\
We thus obtain 
$$\frac{\left(J_0-z_n\right)^{-1}_{j,r}\cdot \left(J_0-z_n\right)^{-1}_{r,k}}{1+\lambda_k\cdot\left(J_0-z_n\right)^{-1}_{r,r}}=A(z_n,\lambda_k)\cdot \left(\phi(z_n)^{|r-j|}-\phi(z_n)^{r+j}\right)\left(\phi(z_n)^{|r-k|}-\phi(z_n)^{r+k}\right)$$
where $A(z_n,\lambda_k)=\left(\frac{1}{\phi(z_n)-1/\phi(z_n)}\right)^2\cdot \frac{1}{1+\lambda_k\frac{1-\phi(z_n)^{2r}}{\phi(z_n)-1/\phi(z_n)}}$.\\
Note that
\begin{equation}\label{eq:asymptoticsphi}\phi(z_n)= \begin{cases}
\frac{x_0 - {\rm i}  \sqrt{4-x_0^2}}{2}\left(1+ \frac{{\rm i } \eta }{n^{\gamma} \sqrt{4-x_0^2} } + \mathcal O(n^{-2\gamma})\right),& \Im \eta>0\\
\frac{x_0+ {\rm i}  \sqrt{4-x_0^2}}{2}\left(1- \frac{{\rm i } \eta }{n^{\gamma} \sqrt{4-x_0^2} } + \mathcal O(n^{-2\gamma})\right), & \Im \eta<0
\end{cases} 
\end{equation}
and therefore $$\left|\phi(z_n)\right|= 1-\frac{\left|\text{Im}\eta\right|}{n^\gamma\sqrt{4-x_0^2}}+o\left(n^{-\gamma}\right)\leq \exp\left(-c\cdot n^{-\gamma}\right)$$
for some $c>0$.\\
Let $n-2mn^{\beta'}\leq r\leq n+2mn^{\beta'}$ then we have the following  \begin{enumerate}
    \item $\phi(z_n)^{2r}=o(1)$ as $n\to \infty$.
    \item $\left|\phi(z_n)-1/\phi(z_n)\right|\xrightarrow[n\to \infty ]{}\sqrt{4-x_0^2}\not=0$.
    \item $\left|\phi(z_n)^{r+k}\right|,\left|\phi(z_n)^{r+j}\right|=\mathcal{O}\left(\exp\left(-c\cdot n^{1-\gamma}\right)\right)$ as $n\to \infty$.
    \item $\left|A(z_n,\lambda_k)\right|\xrightarrow[n\to \infty]{}\frac{1}{4-x_0^2}$.
\end{enumerate}
which means that $G_m^{(n),\beta'}(z_n)$ can be decomposed  as follows (recall $n_k=r$)
$$G_m^{(n),\beta'}(z_n)=\lambda_k\cdot T_m^{(n),\beta'}(z_n)+\mathcal{R}^{(n),\beta'}(z_n)$$ where $T_m^{(n),\beta'}(z_n)$ is the matrix
\beq\label{eq:Tbeta}
T_m^{(n),\beta'}(z_n)_{j,l}=\begin{cases}
    A(z_n,\lambda_k)\phi(z_n)^{|r-k|+|r-j|}\quad  n-2mn^{\beta'}\leq j,l\leq n+2mn^{\beta'}\\
    0 \qquad otherwise
\end{cases}\eeq
and, as $n\to \infty$
$$\mathcal{R}^{(n),\beta'}(z_n)_{j,l}=\begin{cases}
    \mathcal{O}\left(\exp\left(-c\cdot n^{1-\gamma}\right)\right)\quad  n-2mn^{\beta'}\leq j,l\leq n+2mn^{\beta'}\\
    0 \qquad otherwise
\end{cases}.$$
We have 
$$\left\|G^{(n),\beta'}(z_n)\right\|_1\leq\left|\lambda_k\right|\left\|T^{(n),\beta'}(z_n)\right\|_1+\left\|\mathcal{R}^{(n),\beta'}(z_n)\right\|_1 $$ and it can be readily seen that $$\left\|\mathcal{R}^{(n),\beta'}(z_n)\right\|_1\leq\sum_{n-2mn^{\beta'}\leq j,l\leq n+2mn^{\beta'}}\left|\mathcal{R}^{(n),\beta'}(z_n)_{j,l}\right|= o(1)$$ as $n\to \infty$.\\
Therefore, as $n\to \infty$
\beq \label{eq:G<T+o}
\left\|G_m^{(n),\beta'}(z_n)\right\|_1\leq\left|\lambda_k\right|\left\|T_m^{(n),\beta'}(z_n)\right\|_1+o(1) \eeq so by estimating the trace norm of $T_m^{(n),\beta'}(z_n)$, we should obtain an upper bound on the asymptotic behavior of the trace norm of $G_m^{(n),\beta'}(z_n)$.
\subsection{\textbf{Step 3: Estimating the trace norm of a Hankel matrix}}
\hfill\\
Using \eqref{eq:Tbeta}, we decompose $T_m^{(n),\beta'}(z_n)$ into 4 sub-matrices, depending on $\text{sign}(r-l)$ and $\text{sign}(r-j)$ as follows
\beq \label{eq:4Hankels}
\begin{split}
   & T_m^{(n),\beta'}(z_n)\\
   &= \left(\frac{1}{4-x_0^2}+o(1)\right)P_{n-2mn^{\beta'},n+2mn^{\beta'}}\left(
    \begin{array}{c|c}
      \begin{array}{cccc}
      \ddots & \vdots& \vdots&\vdots \\
       \ldots&  q_n^5 & q_n^4 & q_n^3 \\
       \ldots&  q_n^4 & q_n^3 & q_n^2 \\
        \ldots& q_n^3& q_n^2  & q_n
      \end{array} & \begin{array}{cccc}
      \vdots & \vdots& \vdots&\ddots \\
       q_n^3&  q_n^4 & q_n^5 & \ldots \\
       q_n^2&  q_n^3 & q_n^4 & \ldots \\
        q_n& q_n^2& q_n^3  & \ldots
      \end{array}\\
      \hline
      \begin{array}{cccc}
      \ldots & q_n^3& q_n^2&q_n \\
       \ldots&  q_n^4 & q_n^3 & q_n^2\\
       \ldots&  q_n^5 & q_n^4 &q_n^3 \\
        \ddots&  \vdots&  \vdots  & \ddots
      \end{array}& \begin{array}{cccc}
      1 & q_n& q_n^2&\ldots \\
       q_n&  q_n^2 & q_n^3 & \ldots \\
       q_n^2&  q_n^3 & q_n^4 &\ldots \\
        \vdots&  \vdots&  \vdots  & \ddots
      \end{array}
    \end{array}
    \right)P_{n-2mn^{\beta'},n+2mn^{\beta'}}\\
\end{split}\eeq
  where $q_n=\phi(z_n)$ and the '1' is located in the $(r,r)$ entry.\\
\begin{definition}
    An infinite (one sided) matrix $\mathcal{H}$ is called ``Hankel'' if there exists $g:\mathbb{N}\to \mathbb{C}$ so that for every $k,j\in\mathbb{N}$, $\mathcal{H}_{j,k}=g(j+k)$. Namely, $\mathcal{H}$ is constant along its skew diagonals. We say that $\mathcal{H}$ is associated with a symbol $a:\partial\mathbb{D}\to \mathbb{C}$ if $\mathcal{H}_{j,k}=\hat{a}_{j+k-1}$ where $a(z)=\sum_{k=-\infty}^\infty\hat{a}_{k}z^k $, i.e. $\hat{a}_{k}$ is the $k'th$ Fourier coefficient of $a$, and usually we denote $\mathcal{H}(a)$. 
\end{definition}

Examining \eqref{eq:4Hankels}, the unique structure of $T_m^{(n),\beta'}$ can be viewed as a matrix composed of four (partially flipped and shifted) Hankel matrices.\\
We denote $$\mathcal{H}_{z_n}=\left(q_n^{j+k}\right)_{0\leq j,k},$$
so $\mathcal{H}_{z_n}$ is a Hankel matrix.\\
Let $\left\{e_i\right\}_{i=1}^\infty$ the standard basis of $\ell_2\left(\bbN\right)$. We define the following matrices:
\begin{itemize}
    \item The shift operator, $S:\ell_2(\bbN)\to\ell_2(\bbN) $, $S\left(e_i\right)=e_{i+1}$ for any $i\in\mathbb{N}$.
    \item For $j\in \mathbb{N}$ we define $E_j:\ell_2(\bbN)\to\ell_2(\bbN)$, the elementary matrix which preforms $[j/2]$ coordinates interchange 
    \begin{equation*}
      E_j\left(e_i\right)=  \begin{cases}
            e_{j-i+1} \quad 1\leq i \leq j\\
            e_i \quad otherwise
        \end{cases}
    \end{equation*}
    \item $Q_1=1-\mathcal{P}_1$, that is 
    \begin{equation*}
      Q_1\left(e_i\right)=  \begin{cases}
            0 \quad i=1 \\
            e_i \quad otherwise
        \end{cases}
    \end{equation*}
    
\end{itemize}
Using these matrices and viewing \eqref{eq:4Hankels}, it is straightforward that $T_m^{(n),\beta'}(z_n)$ can be written as follows 
\beq
\begin{split}
   & T_m^{\beta'}(z_n)\\
   & = \left(({4-x_0^2})^{-1}+o(1)\right)\cdot\Big(P_{r,n+2mn^{\beta'}}\cdot S^{r}\cdot \mathcal{H}_{z_n} \cdot Q_1\cdot E_{r-1-[n-2mn^{\beta'}]}\cdot P_{n-2mn^{\beta'},r-1}\\
   & \quad+  P_{n-2mn^{\beta'},r}\cdot E_{r-[n-2mn^{\beta'}]}\cdot  \mathcal{H}_{z_n} \cdot Q_1\cdot E_{r-[n-2mn^{\beta'}]}\cdot P_{n-2mn^{\beta'},r}\\
    & \quad+  P_{r,n+2mn^{\beta'}}\cdot E_{r-[n-2mn^{\beta'}]}\cdot  \mathcal{H}_{z_n} \cdot Q_1\cdot S^r \cdot P_{n-2mn^{\beta'},r}\\
    & \quad+  P_{r,n+2mn^{\beta'}}\cdot S^r\cdot \mathcal{H}_{z_n}\cdot S^r \cdot P_{r,n+2mn^{\beta'}}\Big).
\end{split}
\eeq
Using the fact that $S,E_j,Q_1$ and $P_{n_1,n_2}$ are of norm 1 (operator norm), and the fact that we have for any two matrices $A,B$, $\left\|AB\right\|_1\leq \|A\|_1\|B\|_\infty $ we see that there exists $C>0$ which is independent of $n$ so that 
\beq \label{eq:T<H}
\left\|T^{(n),\beta'}(z_n)\right\|_1\leq C\cdot \left\|\mathcal{H}_{z_n}\right\|_1.\eeq
Our last major step towards the finalization of the proof of Theorem $\ref{thm:premain}$ is 
\begin{prop} \label{prop:H_isBigO}
We have that $$\left\|\mathcal{H}_{z_n}\right\|_1=\mathcal{O}\left(n^\gamma\right)$$ as $n\to \infty$. In Particular, $\left\|T^{(n),\beta'}(z_n)\right\|_1=\mathcal{O}\left(n^\gamma\right)$ as $n\to \infty$.
\end{prop}
\begin{proof}
    We define $g_n:[0,\infty) \to \bbC$ by $g_n(x)=\left(\left(\phi(z_n)\right)^{n^\gamma}\right)^{\frac{x}{n^\gamma}}$, where the exponent is chosen with respect to the branch cut of the logarithm along the positive real line. It is easy to see (using \eqref{eq:asymptoticsphi}) that there exist $d,D>0$ independent of $n$ such that $\forall x\geq0$
    \beq\label{eq:|gn|}\left|g_n(x)\right|\leq D\cdot \exp\left(\frac{-d\cdot x}{n^\gamma}\right) \eeq
    and \beq\label{eq:|gn'|}\left|g'_n(x)\right|\leq \frac {D\cdot d}{n^\gamma}\cdot \exp\left(\frac{-d\cdot x}{n^\gamma}\right)=\left|\frac{d}{dx}\left(D\cdot \exp\left(\frac{-d\cdot x}{n^\gamma}\right)\right)\right|.\eeq
      We see that $\left(\mathcal{H}_n\right)_{j,k}=g_n(j+k),$ or, in other words 
        \beq \mathcal{H}_{z_n}=\mathcal{H}\left(\sum_{0\leq k}g_n(k)z^k\right)=\mathcal{H}\left(\sum_{0\leq k}\text{Re}g_n(k)z^k\right)+i\mathcal{H}\left(\sum_{0\leq k}\text{Im}g_n(k)z^k\right) \eeq 
    thus
    \beq \label{eq:Hankel<ReIm}
    \left\|\mathcal{H}_{z_n}\right\|_1\leq \left\|\mathcal{H}\left(\sum_{0\leq k}\text{Re}g_n(k)z^k\right)\right\|_1+\left\|\mathcal{H}\left(\sum_{0\leq k}\text{Im}g_n(k)z^k\right)\right\|_1 .\eeq
    Writing $g_n=\Re g_n+i\Im g_n$, it is trivial that both $\Re g_n, \Im g_n$ are $C_1$ real functions which also satisfy 
    \beq \label{eq:RegImg}
    \left|\text{Re}g_n\right|,\left|\text{Im}g_n\right|\leq \left|g_n\right|
    \eeq
    and
    \beq \label{eq:Reg'Img'}
    \left|\text{Re}g'_n\right|,\left|\text{Im}g'_n\right|\leq \left|g'_n\right|.
    \eeq
     By \cite[Proposition 3.7]{Hankel_trace} one can obtain an upper estimate of the Besov norm of a given power series (for more details about Besov spaces, see \cite[Section 3]{Hankel_trace} and references therein). Based on results in \cite[Section 8]{Hankel_trace_Peller} (see for example \cite[Subsection 3.2, Equation 7]{Hankel_trace}), the Besov norm of a given power series is equivalent to the trace norm of the Hankel matrix whose symbol is the exact same power series. Combining these facts, we obtain $\mathcal{B}>0$, independent of $n$ or $g_n$ such that
    \beq \label{eq:ReHankel}\left\|\mathcal{H}\left(\sum_{0\leq k}\text{Re}g_n(k)z^k\right)\right\|_1\leq \mathcal{B}\cdot \left(\left|\text{Re}g_n(0)\right|+\sqrt{2}\left(A\left(\text{Re}g_n\right)+\sqrt{A\left(\text{Re}g_n\right)\cdot B\left(\text{Re}g_n\right)}\right)\right)\eeq 
    where for a continuously differentiable real function $f$,
    $$A(f)=\left(\sum_{k=0}^\infty \left|f(k)\right|^2\cdot\sum_{k=1}^\infty k^2\left|f(k)\right|^2\right)^{1/4}$$ 
    and
    $$B(f)=\left(\left\|x\cdot f'\right\|_{L_2(1,\infty)}\cdot \left\|x^2\cdot f'\right\|_{L_2(1,\infty)}\right)^{1/2}.$$ 
    We remark that \eqref{eq:ReHankel} holds if the real part is replaced with the imaginary part.\\
    Using inequalities \eqref{eq:|gn|},\eqref{eq:|gn'|}, \eqref{eq:RegImg},\eqref{eq:Reg'Img'} we get 
    $$A\left(\text{Re}g_n\right),A\left(\text{Im}g_n\right)\leq A\left(g_n\right)\leq A\left(D\exp\left(\frac{-d(\cdot)}{n^\gamma}\right)\right)$$ 
    and
    $$B\left(\text{Re}g_n\right),B\left(\text{Im}g_n\right)\leq B\left(g_n\right)\leq B\left(D\exp\left(\frac{-d(\cdot)}{n^\gamma}\right)\right).$$
    Therefore, denoting $h_n(x)=D\exp\left(\frac{-dx}{n^\gamma}\right)$, by \eqref{eq:Hankel<ReIm} 
    \beq \label{eq:Hankel<h_n}
    \left\|\mathcal{H}_{z_n}\right\|_1\leq2\mathcal{B}\cdot \left(\left|h_n(0)\right|+\sqrt{2}\left(A\left(h_n\right)+\sqrt{A\left(h_n\right)\cdot B\left(h_n\right)}\right)\right).
    \eeq 
    In order to evaluate the RHS of \eqref{eq:Hankel<h_n} we use for $\left|q\right|<1$:
    $$\sum_{k=0}^\infty q^k=\frac{1}{1-q}$$ 
    $$\sum_{k=1}^\infty k\cdot q^{k-1}=\frac{1}{(1-q)^2}$$ and
    $$\sum_{k=1}^\infty k\cdot(k-1) q^{k-2}=\frac{2}{(1-q)^3}.$$ 
    Plugging $q=\left|h_n(1)\right|^2$, we note that $\left|h_n(k)\right|^2=q^k$ and $\left|h_n(1)\right|^2=1-\frac{2d}{n^\gamma}+o\left(n^{-\gamma}\right)$ as $n\to \infty$, therefore we have $$\sum_{k=1}^\infty \left|h_n(k)\right|^2\cdot\sum_{k=1}^\infty k^2\left|h_n(k)\right|^2=\mathcal{O}\left({n^{4\gamma}}\right)$$ as $n \to \infty$, which means $A\left(h_n\right)=\mathcal{O}\left({n^\gamma}\right)$ as $n\to \infty$.\\
    Moving to estimating $B\left(h_n\right)$, we see that using $\left\|x\cdot h'_n(x)\right\|_2=D\cdot d\cdot  \left\|\frac{x}{n^\gamma}\exp\left(\frac{-dx}{n^\gamma}\right)\right\|_2$  
    \beq
    \int_0^\infty \left|\frac{x}{n^\gamma}\exp\left(-d\cdot\frac{x}{n^\gamma}\right)\right|^2\text{d}x\underset{t=x/n^\gamma}{=} n^\gamma\int_0^\infty \left|t\exp\left(-d\cdot t\right)\right|^2\text{d}t=\mathcal{O}\left(n^\gamma\right)
    \eeq
    as $n\to \infty$, hence $\left\|x\cdot h'_n(x)\right\|_2=\mathcal{O}\left(n^{\gamma/2}\right)$.\\
    Moreover, we see $\|x^2\cdot h'_n(x)\|_2=D\cdot d\cdot  \|\frac{x^2}{n^\gamma}\exp\left(\frac{-dx}{n^\gamma}\right)\|_2$.
    When computing the integral as above, one can see $$\left\|x^2\cdot h'_n(x)\right\|_2=\mathcal{O}\left(n^{\frac{3}{2}\gamma}\right),$$ thus $B\left(h_n\right)=\mathcal{O}\left(n^\gamma\right)$ as $n\to \infty$.\\
    We thus get by \eqref{eq:Hankel<h_n},
    $$\left\|\mathcal{H}_{z_n}\right\|_1\leq \mathcal{O}\left(n^\gamma\right)+\sqrt{\mathcal{O}\left(n^\gamma\right)\cdot \mathcal{O}\left(n^\gamma\right)}$$
    as $n\to \infty$.
    This obviously implies
$$\left\|\mathcal{H}_{z_n}\right\|_1=\mathcal{O}\left(n^\gamma\right)$$
    as $n\to \infty.$\\
    We see that due to \eqref{eq:T<H} we immediately get 
    $$\left\|T^{(n),\beta'}(z_n)\right\|_1=\mathcal{O}\left(n^\gamma\right),$$ 
    this completes the proof.
\end{proof}
We are now ready to prove Theorem \ref{thm:premain}. 
\begin{proof}
Let $0<\gamma<1$, let $\gamma<\beta<1$, and let $\text{d}\mu_0(x)=\frac{\sqrt{4-x^2}}{2\pi}\text{d}x$ and $\mu$ such that $J_0$ and $J$, the corresponding Jacobi matrices are $\beta$-sparse perturbation of each other. Let $N\in\mathbb{N}$ and $\eta_1\ldots,\eta_N\in\mathbb{C}$ so that $\text{Im}\eta_j\not=0$ and $c_1\ldots,c_N\in\mathbb{R}$. We define $$f(x)=\sum_{j=1}^N\frac{c_j}{x-\eta_j}.$$ By \eqref{eq:diff_of_cumulants_linearcomb} we have
\begin{equation} 
\begin{split}
    & \left|C^{\mu_0}_m\left(X^{(n)}_{\gamma,x_0}(f)\right)-C^{\mu}_m\left(X^{(n)}_{\gamma,x_0}(f)\right)\right| \\
    &\quad \leq \frac{C(m,\beta')}{n^\gamma}\sum_{j=1}^N|c_j|\left\|P_{n-2mn^{\beta'},n+2mn^{\beta'}}\left(J_0-z_{n,j}\right)^{-1}-\left(J-z_{n,j}\right)^{-1}P_{n-2mn^{\beta'},n+2mn^{\beta'}}\right\|_1.
\end{split}
\end{equation}
By Proposition \ref{CumulantJvsCumulantH} we have 
\beq \begin{split}
    &\frac{C(m,\beta')}{n^\gamma}\sum_{j=1}^N|c_j|\left\|P_{n-2mn^{\beta'},n+2mn^{\beta'}}\left(J_0-z_{n,j}\right)^{-1}-\left(J-z_{n,j}\right)^{-1}P_{n-2mn^{\beta'},n+2mn^{\beta'}}\right\|_1\\
    &\quad \leq \frac{C(m,\beta')}{n^\gamma}\sum_{j=1}^N|c_j|\left\|P_{n-2mn^{\beta'},n+2mn^{\beta'}}\left(H^{n,\beta}_0-z_{n,j}\right)^{-1}-\left(H^{n,\beta}-z_{n,j}\right)^{-1}P_{n-2mn^{\beta'},n+2mn^{\beta'}}\right\|_1+o(1) \\
\end{split}
\eeq 
as $n\to \infty$, and by Lemma \ref{lemmarjk}
\beq \begin{split}
    &\frac{C(m,\beta')}{n^\gamma}\sum_{j=1}^N|c_j|\left\|P_{n-2mn^{\beta'},n+2mn^{\beta'}}\left(H^{n,\beta}_0-z_{n,j}\right)^{-1}-\left(H^{n,\beta}-z_{n,j}\right)^{-1}P_{n-2mn^{\beta'},n+2mn^{\beta'}}\right\|_1 \\
    &\quad \leq \frac{C(m,\beta')}{n^\gamma}\sum_{j=1}^N|c_j|\left\|G^{(n),\beta'}(z_{n,j})\right\|_1+o(1)
\end{split}
\eeq 
where we have by \eqref{eq:G<T+o}  
\begin{equation*}
    \left\|G^{(n),\beta'}(z_{n,j})\right\|_1\leq \left|\lambda_{[n-2mn^\beta,n+2mn^\beta]}\right|\left\|T^{(n),\beta'}(z_{n,j})\right\|_1+o(1)  
\end{equation*}
where $\lambda_{[n-2mn^\beta,n+2mn^\beta]}$ is $\lambda_k$ if $n_k\in [n-2mn^\beta,n+2mn^\beta]$ for some $k$, 0 else. $\lambda_{[n-2mn^\beta,n+2mn^\beta]}$ is well defined (for a large enough $n$) by the $\beta$-spacing condition. Therefore
\beq \begin{split}
    &\frac{C(m,\beta')}{n^\gamma}\sum_{j=1}^N|c_j|\left\|G^{(n),\beta'}(z_{n,j})\right\|_1\\
    &\quad\leq\frac{C(m,\beta')\left|\lambda_{[n-2mn^\beta,n+2mn^\beta]}\right|}{n^\gamma}\sum_{j=1}^N|c_j|\left\|T^{(n),\beta'}(z_{n,j})\right\|_1+o(1) 
\end{split}
\eeq 
which by Proposition \ref{prop:H_isBigO}, as $n\to \infty$ 
\beq \begin{split}
    &\frac{C(m,\beta')\left|\lambda_{[n-2mn^\beta,n+2mn^\beta]}\right|}{n^\gamma}\sum_{j=1}^N|c_j|\left\|T^{(n),\beta'}(z_{n,j})\right\|_1\\
    &\quad= \frac{\left|\lambda_{[n-2mn^\beta,n+2mn^\beta]}\right|}{n^\gamma}\mathcal{O}\left(n^\gamma\right).\\
\end{split}
\eeq
Finally, we thus get 
\begin{equation} 
\begin{split}
    & \left|\mathcal{C}^{(n),\mu_0}_m\left(X^{(n)}_{\gamma,x_0}(f)\right)-\mathcal{C}^{(n),\mu}_m\left(X^{(n)}_{\gamma,x_0}(f)\right)\right| \\
    &\quad \leq \frac{\left|\lambda_{[n-2mn^\beta,n+2mn^\beta]}\right|}{n^\gamma}\mathcal{O}\left(n^\gamma\right)+o(1).
\end{split}
\end{equation}
as $n\to \infty.$\\
Recall $\lambda_k\to 0$ as $k\to \infty$, therefore $\lambda_{[n-2mn^\beta,n+2mn^\beta]}\xrightarrow[n\to\infty]{}0$, which means 
\beq
\left|\mathcal{C}^{(n),\mu_0}_m\left(X^{(n)}_{\gamma,x_0}(f)\right)-\mathcal{C}^{(n),\mu}_m\left(X^{(n)}_{\gamma,x_0}(f)\right)\right| \xrightarrow[n\to \infty]{}0
\eeq 
this completes the proof. \end{proof}
With Theorem \ref{thm:premain} in hand, we are ready to prove Theorem \ref{thm:main}
\begin{proof}
Let $0<\gamma<1$ and $0<\gamma<\beta<1$. Let $x_0\in(-2,2)$ and let $\mu_0,\mu,J_0$ and $J$ be as in Theorem \ref{thm:main}. \\
We observe that both $J_0$ and $J$ are self adjoint, real matrices, therefore $\left((J-z_n)^{-1}\right)^*=(J-\overline{z_n})^{-1}$ which means that 
\beq 
\frac{1}{n^\gamma}\left\|P_{n-2mn^{\beta'},n+2mn^{\beta'}}\left(\left(J_0-\overline{z_n}\right)^{-1}-\left(J-\overline{z_n}\right)^{-1}\right)P_{n-2mn^{\beta'},n+2mn^{\beta'}}\right\|_1\xrightarrow[n\to \infty]{}0\eeq 
as well.\\
Note that for $$\mathfrak{f}(x)=\text{Im}\sum_{j=1}^N\frac{c_j}{x-\eta_j}$$ (where the parameters are as in \eqref{eq:f_poissonkernel}) we have 
$$\mathfrak{f}\left(n^\gamma \left(J-x_0\right)\right)=\frac{1}{2in^\gamma}\sum_{j=1}^Nc_j\left(\left(J-z_{n,j}\right)^{-1}-\left(J-\overline{z_{n,j}}\right)^{-1}\right)$$ with $z_{n,j}$ as in \eqref{eq:diff_of_cumulants_linearcomb}. 
So by Proposition \ref{thm:comp_cumulants} 
\begin{equation} 
\begin{split}
    & \left|\mathcal{C}^{(n),\mu_0}_m\left(X^{(n)}_{\gamma,x_0}(\mathfrak{f})\right)-\mathcal{C}^{(n),\mu}_m\left(X^{(n)}_{\gamma,x_0}(\mathfrak{f})\right)\right| \\
    &\quad \leq \frac{C(m,\beta')}{n^\gamma}\sum_{j=1}^N|c_j|\Bigg(\left\|P_{n-2mn^{\beta'},n+2mn^{\beta'}}\left(J_0-z_{n,j}\right)^{-1}-\left(J-z_{n,j}\right)^{-1}P_{n-2mn^{\beta'},n+2mn^{\beta'}}\right\|_1\\
    &\qquad\qquad\qquad\qquad\quad +\left\|P_{n-2mn^{\beta'},n+2mn^{\beta'}}\left(J_0-\overline{z_{n,j}}\right)^{-1}-\left(J-\overline{z_{n,j}}\right)^{-1}P_{n-2mn^{\beta'},n+2mn^{\beta'}}\right\|_1\Bigg)
\end{split}
\end{equation}
which by Theorem \ref{thm:premain}, goes to 0 as $n\to \infty$. We have hence established a mesoscopic universality for linear combinations of functions of form $\text{Im}\frac{1}{x-\eta}$.\\
Having shown that for any $m\geq 2$ $$\left|\mathcal{C}^{(n),\mu_0}_m\left(X^{(n)}_{\gamma,x_0}(\mathfrak{f})\right)-\mathcal{C}^{(n),\mu}_m\left(X^{(n)}_{\gamma,x_0}(\mathfrak{f})\right)\right| \xrightarrow[n\to \infty]{}0$$ the rest of the proof for $f\in C_c^1\left(\mathbb{R}\right)$ follows precisely as in \cite[Section 5]{BreuerDuitsMeso}. The main ingredient of the proof in \cite{BreuerDuitsMeso} is a density argument that can be applied only when both $\mathcal{C}^{(n),\mu_0}_2\left(X^{(n)}_{\gamma,x_0}(\mathfrak{g})\right)$ and $\mathcal{C}^{(n),\mu}_2\left(X^{(n)}_{\gamma,x_0}(\mathfrak{g})\right)$ are bounded (this is \cite[Proposition 5.2]{BreuerDuitsMeso}), where $\mathfrak{g}(x)=\frac{1}{x-i}$. Indeed, in our case, $\mu_0$ with $x_0\in(-2,2)$ satisfy the conditions of \cite[Proposition 5.2]{BreuerDuitsMeso} and hence $\text{Var}_{\mu_0}\left(X^{(n)}_{\gamma,x_0}(\mathfrak{g})\right)$ is bounded, which is the same as saying that $\mathcal{C}^{(n),\mu_0}_2\left(X^{(n)}_{\gamma,x_0}(\mathfrak{g})\right)$ is. We know, by Theorem \ref{thm:premain} that $$\left|\mathcal{C}^{(n),\mu_0}_2\left(X^{(n)}_{\gamma,x_0}(\mathfrak{g})\right)-\mathcal{C}^{(n),\mu}_2\left(X^{(n)}_{\gamma,x_0}(\mathfrak{g})\right)\right| \xrightarrow[n\to \infty]{}0$$
hence $\text{Var}_{\mu}\left(X^{(n)}_{\gamma,x_0}(\mathfrak{g})\right)$ is bounded as well. From here, the proof follows the same steps as the proof given in \cite[Section 5]{BreuerDuitsMeso}. This completes the proof of Theorem \ref{thm:main}.
\end{proof}

\end{document}